\documentclass[conference]{IEEEtran}
\IEEEoverridecommandlockouts
\usepackage{cite}

\ifCLASSINFOpdf
  \usepackage[pdftex]{graphicx}
  \graphicspath{{../pdf/}{../jpeg/}}
  \DeclareGraphicsExtensions{.pdf,.jpeg,.png}
\else
  \usepackage[dvips]{graphicx}
  \graphicspath{{../eps/}}
  \DeclareGraphicsExtensions{.eps}
\fi
\usepackage{amsmath}
\usepackage{amssymb}
\usepackage{xcolor}
\usepackage{algorithm}
\usepackage{algpseudocode}
\usepackage{url}

\begin{document}
%
\title{Fast Tuning of Pinching Antenna-Assisted ISAC Systems via User Clustering}
%
%
%


\author{
\IEEEauthorblockN{Ruizhi Niu}
\IEEEauthorblockA{
Department of Electrical and Computer Engineering\\
University of Toronto\\
\textit{ruizhi.niu@mail.utoronto.ca}
}
\and
\IEEEauthorblockN{Ali Bereyhi}
\IEEEauthorblockA{
Department of Electrical and Computer Engineering\\
University of Toronto\\
\textit{ali.bereyhi@utoronto.ca}
}
\thanks{This work was supported by the Natural Sciences and Engineering Research Council of Canada (NSERC).}
}

\newtheorem{remark}{Remark}
\newtheorem{theorem}{Theorem}

\maketitle

\begin{abstract}
Pinching-antenna systems (PASSs) offer a great degree of flexibility for integrated sensing and communications (ISAC) by adapting the pinching locations for both communication channel reshaping  and sensing-beam steering. Despite that, direct location optimization in PASSs is high-dimensional and non-convex, making their fast reconfiguration difficult in practice. This work proposes a low-complexity PASS-assisted ISAC design based on \textit{user clustering}. Instead of direct optimization of pinching locations, the elements on each waveguide are distributed over multiple segments, where each segment serves a cluster of users or targets. This alternative configuration reduces the location tuning problem to a low-dimensional weight allocation task, which can be efficiently addressed via a gradient based method. We propose a low-complexity iterative algorithm, in which the digital precoder is updated via an equivalent weighted mean squared error minimization, and pinching locations are updated via weight allocation and projection.   Our numerical experiments show that the proposed design can significantly enhance the communication-sensing trade-off of ISAC systems over the benchmark designs with modest computational cost. The results highlight the potential gain of deploying PASS-assisted transceiver in practical ISAC systems.
\end{abstract}


%
\IEEEpeerreviewmaketitle

\section{Introduction}
%
%
%
%

 
\label{sec:intro}
Pinching-antenna systems (PASSs) have recently emerged as 
promising reconfigurable architectures for 
next-generation wireless systems \cite{Ding2025PASSPerspective}. Using movable pinching elements along dielectric waveguides, PASSs turn 
the antenna array geometry into a design variable 
that can be adapted to the propagation 
environment and the spatial distribution of 
users \cite{Ding2025PASSPerspective,Wang2025ModelingPASS}. Compared 
with conventional rigid-array architectures, this 
geometric degree of freedom offers a new lever for spatial processing. A natural paradigm for the integration of PASSs into wireless networks is integrated sensing and communication (ISAC). 
ISAC consolidates data transmission and target sensing 
on common waveforms and resources. The designs in ISAC simultaneously shape the user-facing channels and the radar illumination pattern. Characterizing 
the resulting trade-off is a central question in the 
ISAC literature. Due to their spatial flexibility, PASSs are of particular interest for ISAC \cite{ISAC1, ISAC2}. In fact, the pinching-element geometry enables 
a hybrid spatial and digital beamforming that is highly flexible. 
Despite its potential, a \textit{practical} PASS-enabled ISAC design remains open: the geometric degrees of freedom offered by PASS deliver substantial gains \textit{only if} the joint tuning of pinching locations and digital beamforming is performed at the system update rate, i.e., once within a  coherence time. Considering the non-convex nature of the underlying design problem, this is a challenging design problem. Motivated by this, this work develops a low-complexity algorithm whose per-iteration cost is dominated by closed-form operations. 

\subsection{Related Work}
\label{subsec:related}
Recent work has explored the deployment of PASSs for various functionalities in wireless systems. Joint digital precoding and pinching-location tuning is addressed in \cite{11162401,bereyhi2026MIMOPASS}, demonstrating substantial gains over fixed-location architectures. Deploying PASSs for non-orthogonal multiple access (NOMA) is studied in \cite{NomaPASS}. Dynamic and static location tuning in PASSs for energy-efficient transmission is investigated in \cite{asaad2026dynamicstaticenergyefficient}, highlighting the importance of updating geometry at the network-coherence rate. 
Several recent lines of work have investigated PASS-assisted ISAC. The study in \cite{ISAC1} develops a hybrid design for ISAC, in which the pinching locations are optimized to maximize the illuminated power. The work in \cite{ISAC2} studies the impact of tunable pinching locations on the achievable sensing-communication rate region. 

The most related study to this work is \cite{asaad2026low}, in which the \textit{bi-partitioning} trick is utilized to design a low-complexity PASS-assisted ISAC system with optimal sensing-communication trade-off. While related, the proposed design in \cite{asaad2026low} relies on a single-user single-waveguide PASS geometry and does not directly extend to \textit{multiuser} MIMO-PASS settings. 
In this work, we bridge this gap by considering a generic multiuser PASS-assisted ISAC setting in which the access point deploys multiple transmit and receive waveguides. We invoke user clustering, as an extension of the bi-partitioning trick, to obtain 
a low-dimensional parameterization of pinching locations. The clusters determine a small set of projection points on the waveguides, and the design variable reduces to the number of elements allocated to each point. 
Using this trick, the aperture split between communication and sensing is tuned by a single low-dimensional allocation vector, which can be efficiently optimized by a simple gradient-based algorithm. 

\subsection{Contributions}
We develop a class of low-complexity algorithms for hybrid precoding in PASS-assisted ISAC systems. The main
contributions are five-fold: 
($i$) We formulate the joint design 
as a multi-objective optimization problem that simultaneously maximizes the communication sum-rate and sensing mutual information. We invoke scalarization to approximate the Pareto optimal solution. 
($ii$) Invoking alternating-optimization (AO), we design an iterative scheme to approximate the joint design. The solution is computed by iterating between the digital and analog precoding tasks. 
($iii$) Using the rank-one factorization of the sensing channel, we reformulate the marginal digital precoding as a unified weighted minimum mean-square-error (WMMSE) problem. We invoke the decoupling between communication and sensing objectives to transform the problem into a quadratic surrogate maximization, for which 
a semi-closed-form solution is determined via Sherman-Morrison rank-one correction. 
($iv$) To address the location optimization task, we exploit the spatial properties of the PASS system. We invoke the $k$-means clustering algorithm to transform the location optimization task into a variational weight-allocation problem, which is efficiently solved using the projected gradient descent algorithm. 
The resulting continuous allocation is then mapped to discrete element positions. 
($v$) We validate our design through numerical experiments. 
Our results show that the proposed PASS-assisted ISAC design can significantly expand the achievable sensing-communication rate region at a practically-tolerable  computational cost.


\paragraph*{Notation}
Bold lowercase and uppercase letters denote vectors and 
matrices, respectively. Superscripts $(\cdot)^{T}$, 
$(\cdot)^{*}$, and $(\cdot)^{H}$ denote transpose, 
conjugate, and conjugate transpose, respectively. The complex plane is denoted by $\mathbb{C}$, and $\mathcal{CN}(\boldsymbol{\mu},\boldsymbol{\Sigma})$ denotes a complex Gaussian distribution with mean $\boldsymbol{\mu}$ and covariance $\boldsymbol{\Sigma}$. The trace is shown by $\mathrm{tr}(\cdot)$. 
We use $\Pi_{\mathcal{S}}(\cdot)$ to denote the 
Euclidean projection onto the set~$\mathcal{S}$.

\section{System Model}
\label{sec:system_model}
Consider a downlink PASS-aided access point (AP) that serves $K$ users and senses a 
single-point target. Let the location of the $k$-th user be $\mathbf{u}_k = [x^{u}_k,\,y^{u}_k,\,z^{u}_k]$, 
for $k = 1,\dots,K$, and the target location be $\mathbf{u}^{t} = [x^{t},\,y^{t},\,z^{t}]$. 
The AP consists of a transmit and a receive PASS. The transmit PASS is equipped with $M_{\mathrm{tx}}$ 
dielectric waveguides extending along the $x$-axis at a 
fixed altitude $h$. The $m_{\mathrm{tx}}$-th waveguide 
has length $L^{\mathrm{tx}}_{m_{\mathrm{tx}}}$ and is located on the $y$-axis at $y^{\mathrm{tx}}_{m_{\mathrm{tx}}} = 
(m_{\mathrm{tx}}-1)d$, for some inter-waveguide spacing $d$. The receive PASS consists of $M_{\mathrm{rx}}$ waveguides of length $L^{\mathrm{rx}}_{m_{\mathrm{rx}}}$ deployed in 
identical configuration, separated from the transmitting block by a gap $D_{\mathrm{gap}}$ on the $y$-axis. We assume that the transmit and receive blocks are sufficiently isolated, such that the residual leakage from the transmit PASS into the receive PASS is negligible. 

Each transmit and receive waveguide 
carries $N_{\mathrm{tx}}$ and $N_{\mathrm{rx}}$ adjustable
pinching elements, respectively. The $n_{\mathrm{tx}}$-th pinching element on the $m_{\mathrm{tx}}$-th transmit waveguide is located at distance $\ell^{\mathrm{tx}}_{m_{\mathrm{tx}},n_{\mathrm{tx}}} \in 
[0,\,L^{\mathrm{tx}}_{m_{\mathrm{tx}}}]$, and hence its coordinate is
\begin{equation}
\mathbf{v}^{\mathrm{tx}}_{m_{\mathrm{tx}},n_{\mathrm{tx}}} 
= \big[\ell^{\mathrm{tx}}_{m_{\mathrm{tx}},
n_{\mathrm{tx}}},\; y^{\mathrm{tx}}_{m_{\mathrm{tx}}},\; 
h\big].
\label{eq:tx_coord}
\end{equation}
The receive element coordinate 
$\mathbf{v}^{\mathrm{rx}}_{m_{\mathrm{rx}},
n_{\mathrm{rx}}}$ for a receive waveguide is defined analogously. 
Aggregating all element positions yields the location vector
\begin{align}
\boldsymbol{\ell}^{i} &\triangleq 
\big[\ell^{i}_{1,1},\dots,
\ell^{i}_{1,N_{i}},\dots,
\ell^{i}_{M_{i},1},\dots,
\ell^{i}_{M_{i},N_{i}}\big] 
\label{eq:ltx_vec}
\end{align}
for $i \in \{\mathrm{tx}, \mathrm{rx}\}$, which collects all analog design variables. 

Let $\mathbf{s}[t] \in 
\mathbb{C}^{K}$ collect the encoded information symbols of the users at time $t$. 
The AP determines the transmit signal from $\mathbf{s}[t]$ via linear precoding,\footnote{Note that a dedicated sensing signal is ignored, as the system performs in \textit{downlink} mode; see for example \cite{asaad2026low}.} i.e.,
\begin{equation}
\mathbf{x}[t] = \mathbf{W}\mathbf{s}[t], 
\label{eq:tx_signal}
\end{equation}
where $\mathbf{W}\in\mathbb{C}^{M_{\mathrm{tx}}\times K}$ is 
the precoding matrix whose $k$-th column 
$\mathbf{w}_k$ represents the beamforming vector for user $k$. We assume that the total transmit power is constrained by $P$. With normalized uncorrelated information symbols, this means that the beamforming matrix satisfies $\mathrm{tr}(\mathbf{W}\mathbf{W}^H) \le P$.

\subsection{Communication Setting}
With line-of-sight propagation, the channel 
coefficient from the $n_i$-th element on the $m_i$-th 
waveguide, for $i \in \{\mathrm{tx}, \mathrm{rx}\}$, to a given location $\mathbf{u}$ is given by \cite{bereyhi2026MIMOPASS}
\begin{align}
   \hspace{-1mm}
    g^{i}_{m_i,n_i}(\mathbf{u}) 
    \hspace{-1mm}
    = 
    \hspace{-1mm}
    \frac{\beta\varrho_i}{D^{i}_{m_i,n_i}(\mathbf{u})}
        e^{ -j\kappa \bigl(D^{i}_{m_i,n_i}(\mathbf{u}) + n_{\mathrm{eff}}\ell^{i}_{m_i,n_i}\bigr) },
        \label{eq:elem_rx}
\end{align}
where $\varrho_\mathrm{rx} = 1$ and $\varrho_\mathrm{tx} = 1/\sqrt{N_{\mathrm{tx}}}$, $D^{i}_{m,n}(\mathbf{u}) = 
\|\mathbf{v}^{i}_{m_i,n_i} - \mathbf{u}\|$ 
is the propagation 
distance, $\kappa = 2\pi/\lambda_{0}$ is the
wavenumber with $\lambda_0$ being the wavelength, $n_{\mathrm{eff}}$ is the effective refractive 
index of the waveguide, and $\beta$ captures the 
 path-loss. The channel from the $m_i$-th waveguide to location $\mathbf{u}$ is obtained by aggregating the element contributions along the waveguide, i.e., 
 $g^{i}_m(\mathbf{u}) = \sum_{n} g^{i}_{m_i,n}(\mathbf{u})$. 
Collecting the channels associated with the 
transmit PASS yields the downlink channel 
\begin{equation}
\mathbf{g}_k(\boldsymbol{\ell}^{\mathrm{tx}}) = 
\big[g^{\mathrm{tx}}_{1}(\mathbf{u}_k),\dots,
g^{\mathrm{tx}}_{M_{\mathrm{tx}}}(\mathbf{u}_k)\big]^T ,
\label{eq:user_channel}
\end{equation}
to user $k$, which depends on the pinching locations. 

The received signal at user $k$ 
is therefore given by
\begin{equation}
y_k[t] = \mathbf{g}_k^T\mathbf{W}\mathbf{s}[t] + 
\varepsilon_k[t],
\label{eq:user_rx}
\end{equation}
for additive Gaussian noise $\varepsilon_k[t]\sim\mathcal{CN}(0,\sigma_u^{2})$. The signal-to-noise-and-interference (SINR) at user $k$ is thus given by
\begin{align}
    \mathrm{SINR}_k(\mathbf{W},\boldsymbol{\ell}^{\mathrm{tx}}) &= 
\frac{|\mathbf{g}_k^T\mathbf{w}_k|^{2}}
{\sum_{j\ne k}|\mathbf{g}_k^T\mathbf{w}_j|^{2} + 
\sigma_u^{2}}.
\end{align}
We can hence write the achievable rate to user $k$ as 
\begin{align}
    R_k(\mathbf{W},\boldsymbol{\ell}^{\mathrm{tx}}) &= \log_2(1+\mathrm{SINR}_k(\mathbf{W},\boldsymbol{\ell}^{\mathrm{tx}})),
\end{align}
and the weighted sum-rate as 
\begin{align}
    R_c(\mathbf{W},\boldsymbol{\ell}^{\mathrm{tx}}) = \sum_{k=1}^{K}\lambda_k R_k(\mathbf{W},\boldsymbol{\ell}^{\mathrm{tx}}),
\end{align}
for some weights $\lambda_k\ge 0$. In the sequel, we take the sum-rate expression as the metric to quantify the communication performance of the ISAC system. 

\subsection{Sensing Setting}
Let us define the steering vector to (and from) the target as 
\begin{align}
    \mathbf{a}^{i}(\mathbf{u}^{t}) = 
[g^{i}_{1}(\mathbf{u}^{t}),\dots,
g^{i}_{M_{i}}(\mathbf{u}^{t})]^T,
\end{align}
for $i \in \{\mathrm{tx}, \mathrm{rx}\}$.
The sensing channel is then described as $\mathbf{y}_{\mathrm{rx}} [t] = \mathbf{G}(\mathbf{u}^{t}) \mathbf{x} [t] + \varepsilon_{\mathrm{rx}}[t]$, where $ \varepsilon_{\mathrm{rx}}[t]\sim\mathcal{CN}(0,\sigma_{\mathrm{rx}}^{2})$, and 
\begin{equation}
\mathbf{G}(\mathbf{u}^{t}) = 
\mathbf{a}^{\mathrm{rx}}(\mathbf{u}^{t})(\mathbf{a}^{\mathrm{tx}}(\mathbf{u}^{t}))^{T}.
\label{eq:G_rankone}
\end{equation}
Assuming a frame of length $T$, the sensing mutual information (SMI) is upper-bounded as \cite{asaad2026low}
\begin{equation}
    I_s \le \frac{1}{T} \log_{2}\left( 1 + T \frac{\rho_t^2}{\sigma_{\text{rx}}^2} \operatorname{tr}(\mathbf{G}(u^t) \mathbf{W} \mathbf{W}^H \mathbf{G}(u^t)^H) \right),
\end{equation}
where $\rho_t^{2}$ is the target reflection variance. 
We denote this upper bound by $I_s^{\mathrm d} (\mathbf{W},\boldsymbol{\ell}^{\mathrm{tx}},
\boldsymbol{\ell}^{\mathrm{rx}})$ and adopt it as the sensing performance metric in the sequel. 
Noting that $\mathrm{tr}(\mathbf{G}
\mathbf{W}\mathbf{W}^H\mathbf{G}^H) = 
\|\mathbf{a}^{\mathrm{rx}}\|^{2}\, 
\|\mathbf{W}^H\mathbf{a}^{\mathrm{tx}\,*}\|^{2}$ and defining $\eta_s \triangleq 
T\rho_t^{2}\sigma_{\mathrm{rx}}^{-2}\,
\|\mathbf{a}^{\mathrm{rx}}\|^{2}$ and $\tilde{\mathbf{a}} 
\triangleq \sqrt{\eta_s}\,\mathbf{a}^{\mathrm{tx}\,*}$, the sensing metric  is compactly written as
\begin{equation}
I_s^{\mathrm d}(\mathbf{W},\boldsymbol{\ell}^{\mathrm{tx}},
\boldsymbol{\ell}^{\mathrm{rx}}) = 
\frac{1}{T}\log_{2}\!\big(1 + 
\|\mathbf{W}^H\tilde{\mathbf{a}}\|^{2}\big).
\label{eq:Is_compact}
\end{equation}

\section{PASS-Assisted ISAC Design}
\label{sec:joint_design}

The PASS-enabled ISAC transmitter described in 
Section~\ref{sec:system_model} pursues two functionally 
distinct objectives through a single physical resource: 
delivering high sum-rate, and extracting informative echoes from the sensing 
target. These two objectives are intrinsically coupled. In the sequel, we formulate a hybrid design through the multi-objective optimization framework.

\subsection{Multi-Objective Formulation}
Treating the sum-rate $R_c$ and SMI $I_s^{\mathrm d}$ as two competing performance metrics, the hybrid design is formulated as
\begin{equation*}
\begin{aligned}
(\mathrm{P}_0):\quad &\max_{\mathbf{W},\,
\boldsymbol{\ell}^{\mathrm{tx}},\,
\boldsymbol{\ell}^{\mathrm{rx}}}\;\;
\big[R_c(\mathbf{W},\boldsymbol{\ell}^{\mathrm{tx}}),\;\;
I_s^{\mathrm d}(\mathbf{W},\boldsymbol{\ell}^{\mathrm{tx}},
\boldsymbol{\ell}^{\mathrm{rx}})\big] \\
\;\;\mathrm{s.t.}\;\;
&\mathrm{tr}(\mathbf{W}\mathbf{W}^H)\le P,\;\; 
0\le \ell^{i}_{m,n}\le 
L^{i}_{m}: \;\; i\in \{\mathrm{tx},\mathrm{rx} \}.
\end{aligned}
\label{eq:P0}
\end{equation*}
The solutions to $\mathrm{P}_0$ describe a Pareto front: 
each point on this front corresponds to an allocation of 
the shared aperture between communication and sensing 
that cannot be improved on one objective without 
sacrificing the other.

To convert $\mathrm{P}_0$ into a tractable optimization  
we invoke the scalarization technique: we take a weighted-sum of objectives as a linear surrogate.
Using this technique, the design problem is reformulated as
\begin{equation*}
\begin{aligned}
(\mathrm{P}_1):\quad 
&\max_{\mathbf{W},\,
\boldsymbol{\ell}^{\mathrm{tx}},\,
\boldsymbol{\ell}^{\mathrm{rx}}}\;\;
w_c\,R_c(\mathbf{W},
\boldsymbol{\ell}^{\mathrm{tx}}) + 
w_s\,I_s^{\mathrm d}(\mathbf{W},
\boldsymbol{\ell}^{\mathrm{tx}},
\boldsymbol{\ell}^{\mathrm{rx}}) \\
\;\;\mathrm{s.t.}\;\;
&\mathrm{tr}(\mathbf{W}\mathbf{W}^H)\le P ,\;\; 
0\le \ell^{i}_{m,n}\le 
L^{i}_{m}: \;\; i\in \{\mathrm{tx},\mathrm{rx} \},
\end{aligned}
\label{eq:P1}
\end{equation*}
for some $w_c,w_s\ge 0$ that satisfy $w_c+w_s=1$.
As the weight pair $(w_c,w_s)$ sweeps the unit simplex, 
the optimal value of $\mathrm{P}_1$ produces a sampled trade-off curve that serves as an inner approximation to the true Pareto frontier \cite{MOOP1}.

\subsection{Solution by Alternating Optimization in a Nutshell}
\label{subsec:ao_overview}
The variational problem $\mathrm{P}_1$ is non-convex on three counts: 
($i$) the sum-rate is non-concave in 
$\mathbf{W}$, ($ii$) the SMI is 
non-concave in $\mathbf{W}$, and ($iii$) both 
$\mathbf{g}_k$ and $\mathbf{a}^{\mathrm{tx}}$ depend 
non-linearly on the pinching locations through an exponential kernel. 
A direct solution is thus intractable. We hence approximate the solution using the AO method. 

The AO-based design approximates the solution by iterating between two marginal tasks: ($i$) digital precoding, and ($ii$) pinching location tuning. For precoding, we develop a unified WMMSE reformulation of the subproblem. This enables us to 
find a closed-form update rule. 
For location tuning, we exploit the spatial 
clustering of the users. Specifically, the $K$ 
users are partitioned into $C \leq K$ clusters and each cluster, together with the target, defines a projection point on 
the waveguide axis. This enables us to transform the  objective to a low-dimensional parameterized function whose optimization can be efficiently carried out 
using projected gradient descent. The optimal design is then approximated by iterating between the two subproblems for a few iterations. In the sequel, we go through each subproblem separately.


\section{Digital Precoding via WMMSE}
\label{subsec:precoder}
With fixed pinching locations, all location-dependent quantities
become constant, and the precoding subproblem reads
\begin{equation*}
\mathrm{P}_2:\ \max_{\mathbf{W}}\ w_c R_c(\mathbf{W}) + w_s I_s^{\mathrm d}(\mathbf{W})
\quad \text{s.t.}\ \mathrm{tr}(\mathbf{W} \mathbf{W}^H) \le P. \label{eq:P2}
\end{equation*}
Neither $R_c$ nor $I_s^{\mathrm d}$ is concave in $\mathbf{W}$; thus, direct
optimization is intractable. To approximate the solution, we rewrite the problem in a variational form.

\subsection{Variational WMMSE Problem}
We rewrite $\mathrm{P}_2$ as a WMMSE problem. 
This enables us to form a semi-closed-form $\mathbf{W}$ update. The variational problem is presented in the following theorem.

\begin{theorem}[Equivalent WMMSE Problem]\label{th:1}
The solution to $\mathrm{P}_2$ is equivalently given by the following minimization
\begin{equation}
\min_{\mathbf{W},\,\{c_k,v_k\},\,\mathbf{u},\,c_s,\,v_s}\ F
\quad \text{s.t.}\ \mathrm{tr}(\mathbf{W} \mathbf{W}^H) \le P, \|\mathbf{u}\|=1,
\label{eq:joint_min}
\end{equation}
where $\mathbf{u}\in\mathbb{C}^K$, $c_k,c_s\in\mathbb{C}$, and
$v_k,v_s>0$ for $k=1,\ldots,K$. The objective $F$ is defined as 
\begin{equation}
F = w_c\!\sum_{k=1}^K\!\lambda_k(v_k e_k - \ln v_k)
+ w_s \lambda_s (v_s e_s - \ln v_s), 
\end{equation}
for $\lambda_s = 1/T$ with
\begin{align}
    e_k &= |c_k|^2(\sum_j |\mathbf{g}_k^T \mathbf{w}_j|^2 + \sigma_u^2) - 2\Re\{c_k \mathbf{g}_k^T \mathbf{w}_k\} + 1,\\
    e_s &= |c_s|^2(|\tilde{\mathbf{a}}^H \mathbf{W} \mathbf{u}|^2 + 1)
- 2\Re\{c_s\,\tilde{\mathbf{a}}^H \mathbf{W} \mathbf{u}\} + 1.
\end{align}
\end{theorem}

\begin{IEEEproof}[Proof Sketch]
Introducing a scalar receive equalizer $c_k \in \mathbb{C}$ and a positive weight $v_k > 0$ for user $k$, the classical WMMSE identity concludes that \cite{Shi2011WMMSE}
\begin{equation}
\ln(1+\mathrm{SINR}_k) = \max_{c_k,\,v_k>0}\!\big[\ln v_k - v_k\,e_k(\mathbf{W},c_k) + 1\big],
\label{eq:wmmse_comm}
\end{equation}
where $e_k$ is the per-user MSE as defined in the theorem.
The solution to this problem is given by 
\begin{equation}
(v_k^\star , c_k^\star) = \left( 1 + \mathrm{SINR}_k,  \tfrac{(\mathbf{g}_k^T \mathbf{w}_k)^\ast}{\sum_j |\mathbf{g}_k^T \mathbf{w}_j|^2 + \sigma_u^2} \right).
\label{eq:wmmse_comm_star}
\end{equation}
For the sensing metric, we note that
$I_s^{\mathrm d}$
captures the same form as data rate in $\mathbf{W}$. Applying the
Cauchy-Schwarz identity, i.e.,
\begin{equation}
\|\mathbf{W}^H \tilde{\mathbf{a}} \|^2 = \max_{\|\mathbf{u}\|=1} |\tilde{\mathbf{a}}^H \mathbf{W} \mathbf{u}|^2,
\label{eq:cs}
\end{equation}
it reduces to a scalar effective channel $\tilde{\mathbf{a}}^H \mathbf{W} \mathbf{u}$. Introducing a unit direction $\mathbf{u} \in \mathbb{C}^K$, a scalar
estimator $c_s \in \mathbb{C}$, and a sensing weight $v_s > 0$,
the WMMSE argument yields \cite{Shi2011WMMSE}
\begin{equation}
\ln(1 + \|\mathbf{W}^H \tilde{\mathbf{a}}\|^2) =
\!\!\max_{\substack{\|{\mathbf{u}}\|=1\\c_s,\,v_s>0}}\!\!
\big[\ln v_s - v_s\,e_s + 1\big],
\label{eq:wmmse_sens}
\end{equation}
for  $e_s$ given in the theorem. 
For $\mathbf{W}^H \tilde{\mathbf{a}} \neq 0$,
the solution reads
\begin{equation}
(\mathbf{u}^\star, c_s^\star, v_s^\star) = \left(
\tfrac{\mathbf{W}^H \tilde{\mathbf{a}}}{\|\mathbf{W}^H \tilde{\mathbf{a}}\|},
\tfrac{\|\mathbf{W}^H \tilde{\mathbf{a}}\|}{1 + \|\mathbf{W}^H \tilde{\mathbf{a}}\|^2},
1 + \|\mathbf{W}^H \tilde{\mathbf{a}}\|^2
\right).
\label{eq:wmmse_sens_star}
\end{equation}
In the degenerate case $\mathbf{W}^H \tilde{\mathbf{a}} = 0$, any unit vector $\bf u$ is optimal
and $c_s^\star = 0,\ v_s^\star = 1$.
Combining
\eqref{eq:wmmse_comm} and \eqref{eq:wmmse_sens}, the proof is concluded after a few lines of derivation.
\end{IEEEproof}

\subsection{Solution to the Variational Problem}
The WMMSE identities used in Theorem~\ref{th:1} are exact at primal points. Consequently, block-coordinate descent (BCD) improves the objective of $\mathrm{P}_2$ monotonically. With all auxiliary variables fixed to their marginal optima in \eqref{eq:wmmse_comm_star} and \eqref{eq:wmmse_sens_star},
the subproblem becomes quadratic in $\mathbf{W}$. Let us now define
$\mathbf{\Phi}_c = \!\sum_{k=1}^K\! w_c \lambda_k v_k |c_k|^2\, \mathbf{g}_k^\ast \mathbf{g}_k^T$, the matrix $\mathbf{B}_c$ with columns 
$[\mathbf{B}_c]_{:,k} = w_c \lambda_k v_k c_k^\ast \mathbf{g}_k^\ast$, and scalars 
$\alpha_s = w_s \lambda_s v_s |c_s|^2$ and
$\beta_s = w_s \lambda_s v_s c_s^\ast$.
Introducing a Lagrange multiplier $\mu \ge 0$ for the power
constraint, the KKT stationarity condition yields
\begin{equation}
(\mathbf{\Phi}_c + \mu \mathbf{I})
\mathbf{W}
+ \alpha_s \tilde{\mathbf{a}} \tilde {\mathbf{a}}^H \mathbf{W} \mathbf{u} \mathbf{u}^H
= \mathbf{B}_c + \beta_s \tilde{\mathbf{a}} \mathbf{u}^H.
\label{eq:kkt}
\end{equation}
For $\mu > 0$, we define $\mathbf{A} \triangleq \mathbf{\Phi}_c + \mu \mathbf{I}\succ 0$. Applying Sherman-Morrison identity to invert rank-one perturbation, we have 
\begin{equation}
\mathbf{W}^\star(\mu) = \mathbf{A}^{-1} \mathbf{B}_c + \xi(\mu)\, \mathbf{q}(\mu)\, \mathbf{u}^H,
\label{eq:Wstar}
\end{equation}
where $\mathbf{q}(\mu) = \mathbf{A}^{-1} \tilde{\mathbf{a}}$, and
\begin{equation}
\xi(\mu) = \tfrac{\beta_s - \alpha_s\, \mathbf{q}^H(\mu) \mathbf{B}_c \mathbf{u}}{1 + \alpha_s \gamma(\mu)},
\label{eq:Wstar_aux}
\end{equation}
with $\gamma(\mu) = \tilde{\mathbf{a}}^H \mathbf{q}(\mu)$. Noting that $\alpha_s \ge 0$ and $\gamma(\mu) \ge 0$, we have $1 + \alpha_s \gamma(\mu) \ge 1$ guaranteeing the update is well-defined.

The solution is given by iterating between \eqref{eq:wmmse_comm_star}, \eqref{eq:wmmse_sens_star}, and \eqref{eq:Wstar}. The dual variable $\mu$ is chosen by bi-sectioning to meet the power constraint. 


\section{Location Tuning via Clustering}
\label{subsec:location}
We now solve 
the location tuning subproblem in $\mathrm{P}_1$. 
Starting with the receive-side locations, i.e., $\boldsymbol{\ell}^{\mathrm{rx}}$, the
subproblem admits a single-step closed-form solution:
for any fixed $\mathbf{W}$ and transmit-side configuration, the sensing metric
$I_s^{\mathrm d}$ depends on $\boldsymbol{\ell}^{\mathrm{rx}}$ only through
$\|\mathbf{a}^{\mathrm{rx}}\|^2$, which is maximized when all
$N_{\mathrm{rx}}$ elements on each receive waveguide are aligned with the target by symmetric placement around $x^t$ with spacing $\lambda_g = \lambda_0/n_{\mathrm{eff}}$ \cite{asaad2026low}. As we adopt this placement, $\|\mathbf{a}^{\mathrm{rx}}\|^2$ is treated as a scene-dependent constant while optimizing transmit locations. 

\subsection{Transmit PASS Tuning via User Clustering}
To tune the transmit PASS, we first invoke the $k$-means clustering algorithm to group the users into $C$ distinct clusters. Let
$\mathcal{P} = \{p_0, p_1, \ldots, p_C\}$ be the
set of centroids of clusters projected onto the $x$-axis, with $p_0 = x^t$ being the target
projection, and $\{p_1, \ldots, p_C\}$ the sorted $x$-coordinates of
the cluster centroids. Set $G_c = C+1$. For the $m$-th waveguide, let $\alpha^{(m)} = [\alpha_0^{(m)}, \ldots, \alpha_C^{(m)}]^T$
be an allocation vector in the probability simplex. We approximate 
$N_{m,i} \approx N_{\mathrm{tx}} \alpha_i^{(m)}$ to denote a \textit{continuous} element count assigned to cluster $i$. The closest integer determines the exact number of elements on the waveguide $m$ serving cluster $i$. Collecting these weights, we have allocation matrix
$\boldsymbol{\alpha} \in \mathbb{R}_+^{M_{\mathrm{tx}} \times G_c}$ which belongs to 
\begin{equation}
\mathcal{A} = \Big\{\boldsymbol{\alpha}:\ \alpha_i^{(m)} \ge 0,\
\sum_{i=0}^C \alpha_i^{(m)} = 1,\ \forall m\Big\}.
\label{eq:simplex}
\end{equation}

The elements allocated to each cluster are placed around the
corresponding projection point. We hence model the $N_{m,i}$ elements as occupying the segment $[\bar p_{m,i} - \Delta_{m,i},\ \bar p_{m,i} + \Delta_{m,i}]$ with half-support
\begin{equation}
\Delta_{m,i} = \tfrac{1}{2}(N_{m,i} - 1)_+\,\lambda_g,
\label{eq:halfsupp}
\end{equation}
where $(\cdot)_+ = \max(\cdot, 0)$, and $\bar p_{m,i}$ is the
group center, initialized at $p_i$ and translated if
the segment extends beyond the boundaries of $[0, L_m^{\mathrm{tx}}]$. 
The response of a single element located at $\ell$ on the $m$-th
transmit waveguide follows from \eqref{eq:elem_rx} by setting
$\ell^{\mathrm{tx}}_{m,n} = \ell$. For brevity, we define
$\psi_m(\ell;\mathbf{u}) = \varphi_m(\ell;\mathbf{u})\,
e^{-j\kappa n_{\mathrm{eff}}\ell}$, where
$\varphi_m(\ell;\mathbf{u}) = \tfrac{\beta\varrho_{\mathrm{tx}}}
{D_m(\ell;\mathbf{u})}\,e^{-j\kappa D_m(\ell;\mathbf{u})}$
is the slowly varying free-space envelope and
$D_m(\ell;\mathbf{u}) = \|[\ell,\,y^{\mathrm{tx}}_{m},\,h] - \mathbf{u}\|$.
Since $\kappa n_{\mathrm{eff}}\lambda_g = 2\pi$, consecutive elements
of a group differ in guided phase by exactly $2\pi$ and the whole group
shares a common guided phase. This way, the $\lambda_g$-spaced
elements combine coherently. Averaging the full response over the
segment would instead advance this phase by
$\kappa n_{\mathrm{eff}}\,2\Delta_{m,i} = 2\pi(N_{m,i}-1)$, i.e., by
$N_{m,i}-1$ complete turns. As a result, the segment average nearly vanishes
and the surrogate would report destructive combining, where the array in
fact combines constructively. We therefore hold the guided phase fixed
at the group center and average the envelope only, i.e., the
response of group $i$ at $\mathbf{u}$ reads
\begin{equation}
\bar\psi_{m,i}(\mathbf{u}) = e^{-j\kappa n_{\mathrm{eff}}\bar p_{m,i}}\,
\tfrac{1}{2}\!\int_{-1}^{1}
\varphi_m(\bar p_{m,i} + \xi\Delta_{m,i};\,\mathbf{u})\,d\xi.
\label{eq:psibar}
\end{equation}
For $\Delta_{m,i} = 0$, i.e.\ $N_{m,i} \le 1$, \eqref{eq:psibar} reduces
to the point-source value $\psi_m(\bar p_{m,i}; \mathbf{u})$. Otherwise,
the integral is approximated by a $Q$-point Gauss-Legendre rule
\cite{Golub1969GaussQuad} with $Q=16$, which is chosen to resolve the
residual oscillation of $\varphi_m$ at grazing angles.
Equation \eqref{eq:psibar} anchors the guided phase at the group center
and is therefore a smooth surrogate of the discrete group response rather
than an exact integral representation of it. All reported operating
points are re-evaluated on the exact element-level channels after
discrete recovery. The effective waveguide channel and transmit steering
vector are then given by
\begin{equation}
g_m^{\mathrm{tx}}(\mathbf{u}; \boldsymbol{\alpha})
= \sum_{i=0}^C N_{m,i}\, \bar\psi_{m,i}(\mathbf{u}),
\label{eq:eff_channel}
\end{equation}
from which $\mathbf{g}_k(\boldsymbol{\alpha})$ and
$\mathbf{a}^{\mathrm{tx}}(\boldsymbol{\alpha})$ are obtained.

\subsection{Parameterized Allocation Problem} 

Setting the pinching elements to be located on cluster segments, the location tuning subproblem is transformed to a parameterized weight allocation task, in which the allocation matrix $\boldsymbol{\alpha}$ is determined such that the surrogate objective in $\mathrm{P}_1$ is maximized. The problem is written as
\begin{equation*}
\mathrm{P}_3:\quad \min_{\boldsymbol{\alpha} \in \mathcal{A}}\
f(\boldsymbol{\alpha}) \triangleq -w_c R_c(\boldsymbol{\alpha}) - w_s I_s^{\mathrm d}(\boldsymbol{\alpha}),
\label{eq:Palpha}
\end{equation*}
where $R_c(\boldsymbol{\alpha})$ and $I_s^{\mathrm d}(\boldsymbol{\alpha})$ are the sum-rate and SMI, when parameterized by allocation matrix $\boldsymbol{\alpha}$, respectively. We skip the presentation of the exact form, due to the page limit, and leave it for an extended version of the work.

In $\mathrm{P}_3$, the set $\mathcal{A}$ is convex, but the objective is non-convex. We however note two key properties of $\mathrm{P}_3$: (i) the simplex projection
admits a closed-form sort-and-threshold expression
computable in $O(G_c\log G_c)$ \cite{Wang2013SimplexProj}, and (ii) $\boldsymbol{\alpha}$ has only
$M_{\mathrm{tx}} G_c$ entries. Consequently, the
gradient is cheaply evaluated by forward differences. We hence approximate the solution by a stationary point that is found using projected-gradient method (PGM).

\subsubsection*{PGM-based Solution} 
In PGM, we iteratively update $\boldsymbol{\alpha}$ as
\begin{equation}
\boldsymbol{\alpha}^{(t+1)} = \Pi_{\mathcal{A}}\!\big(
\boldsymbol{\alpha}^{(t)} - s_t\, \widehat{\nabla} f(\boldsymbol{\alpha}^{(t)})\big),
\label{eq:pgm}
\end{equation}
where $\widehat{\nabla} f$ is evaluated by forward differences. The step-size $s_t$ is selected by Armijo backtracking, and the projected-step norm is capped to prevent large jumps across the simplex during the weight sweep. A detailed convergence analysis is omitted due to lack of space and deferred to the extended version of the work.


\subsubsection*{Recovering Discrete Allocations} The continuous
allocation $\alpha^\star$ is converted into a physically realizable PASS configuration in three steps: (i) the continuous element
count $\widehat N_{m,i} = N_{\mathrm{tx}} \alpha_i^{(m)\star}$ is rounded into integer counts $N_{m,i}^{\mathrm{disc}}$ via the largest-remainder method, ensuring
$\sum_i N_{m,i}^{\mathrm{disc}} = N_{\mathrm{tx}}$. To reduce branch switching across neighboring weight points, the rounding target is biased toward the previous reported integer allocation, i.e.,
$\widetilde N_{m,i} =
(1-\rho)\widehat N_{m,i} +
\rho N^{\rm prev}_{m,i}$ for $ 0\leq \rho <1$,
where $N^{\rm prev}_{m,i}$ is the recovered allocation at the previous weight point. When no previous operating point is available, we set $\rho=0$. (ii) The
$N_{m,i}^{\mathrm{disc}}$ elements assigned to projection point
$p_i$ on waveguide $m$ are placed symmetrically around
$\bar p_{m,i}$ as 
\begin{equation}
\ell^{\mathrm{tx}}_{m,\,O_{m,i}+q}
= \bar p_{m,i} + \left(q - \tfrac{N_{m,i}^{\mathrm{disc}} + 1}{2}\right)\lambda_g
\label{eq:place}
\end{equation}
for $q \in \{1,\ldots,N_{m,i}^{\mathrm{disc}}\}$, 
where $O_{m,i} = \sum_{r<i} N_{m,r}^{\mathrm{disc}}$ is the
cumulative offset and $\bar p_{m,i}$ is the boundary-adjusted
group centre defined above. (iii) A residual gap between the continuous surrogate and the deployable
configuration survives any quadrature order, since a continuous segment
and a set of $\lambda_g$-spaced points are different objects rather than
an integral and its approximation. As the integer allocation is a small
object with $M_{\mathrm{tx}}G_c$ entries, we close this gap directly on
the exact channels by a greedy pass: single elements are moved between
projection groups on a waveguide, and a move is retained only if it
strictly improves $w_c R_c + w_s I_s^{\mathrm d}$ with $\mathbf{W}$
re-optimized. This pass acts only on the reported configuration and is
never fed back into the AO loop. After this recovery, the
exact element-level channels are recomputed from the recovered
positions, and the precoder is refined once using the WMMSE update for
performance reporting.

\subsection{Complexity of AO-Based Algorithm}

Let $I_w$ denote the number of inner WMMSE updates in the precoding subproblem, $I_\mu$ the number of bisection steps on the dual variable $\mu$, $I_p$ the number of PGM steps in the location subproblem, and $Q$ the number of Gauss-Legendre nodes used in \eqref{eq:psibar}. 
The precoding subproblem updates the auxiliary variables in $O(K^2 M_{\mathrm{tx}})$, forms $\mathbf{\Phi}_c$ and $\mathbf{B}_c$ in $O(K M_{\mathrm{tx}}^2)$, factorizes $\mathbf{A}=\mathbf{\Phi}_c + \mu \mathbf{I}$ at cost $O(M_{\mathrm{tx}}^3)$ for each trial value of $\mu$, and evaluates \eqref{eq:Wstar} at cost $O(K M_{\mathrm{tx}}^2)$.
The overall cost is $O(I_w(K^2 M_{\mathrm{tx}} + K M_{\mathrm{tx}}^2 + I_\mu(M_{\mathrm{tx}}^3 + K M_{\mathrm{tx}}^2)))$. 
In the location subproblem, the reduced variable $\boldsymbol{\alpha}$ has only $n_{\boldsymbol{\alpha}} = M_{\mathrm{tx}} G_c$ entries. 
One evaluation of the objective in $\mathrm{P}_3$ builds the group responses \eqref{eq:psibar} at the $K$ users and at the target in $O(M_{\mathrm{tx}} G_c Q K)$.
It further evaluates the rate and SMI terms in $O(K^2 M_{\mathrm{tx}})$.
A forward-difference gradient therefore costs $n_{\boldsymbol{\alpha}} + 1$ such evaluations, while the simplex projection adds only $O(M_{\mathrm{tx}} G_c \log G_c)$.
The location update hence scales as $O(I_p M_{\mathrm{tx}}^2 G_c^2 Q K)$ to leading order.

It is worth mentioning two points here: (i) the dimension and cost of the continuous location update do not scale with the element budget $N_{\mathrm{tx}}$. Increasing the number of pinching elements changes the integer counts $N_{m,i}$, but does not change the dimension of $\boldsymbol{\alpha}$. The discrete recovery in \eqref{eq:place} scales as $O(M_{\mathrm{tx}} N_{\mathrm{tx}})$, while the subsequent element-level channel evaluation scales as $O(M_{\mathrm{tx}} N_{\mathrm{tx}} K)$. They are executed once per reported operating point rather than per PGM iteration. (ii) For comparison, a direct forward-difference location update that treats the $M_{\mathrm{tx}} N_{\mathrm{tx}}$ element coordinates as optimization variables requires $M_{\mathrm{tx}} N_{\mathrm{tx}} + 1$ exact objective evaluations per gradient. Considering the dominant element-level channel evaluation cost, this approach scales as $O(I_p M_{\mathrm{tx}}^2 N_{\mathrm{tx}}^2 K)$. In contrast, the proposed parameterization scales as $O(I_p M_{\mathrm{tx}}^2 G_c^2 Q K)$. Hence, for fixed $G_c$ and $Q$, the leading-order complexity reduction grows as $ N_{\mathrm{tx}}^2 / (G_c^2 Q)$.

\section{Numerical Investigation}
\label{sec:numerical_results}

We next validate the proposed design and compare it against the baseline and conventional MIMO ISAC designs. 

\subsection{Simulation Setup and Baseline}
\label{subsec:sim_setup}

Unless otherwise specified, the transmitter employs $M_{\rm tx}=2$ waveguides and the receiver employs a single waveguide, i.e., $M_{\rm rx}=1$. All waveguides are equipped with $N_{\rm tx}=N_{\rm rx}=20$ pinching elements.
We assume that the waveguides are aligned through the $x$-axis with equal length extending over the environment side-length $D_x =20$, i.e., $L_{\rm tx}=L_{\rm rx}=D_x=20$.  The carrier frequency is set to
$f_c=28$~GHz with effective refractive index $n_{\mathrm{eff}}=1.4$. 
The
transmit power budget is $P=1$, the downlink noise variance is
$\sigma_u^2 = 0.05$, and the sensing echo noise variance is
$\sigma_{\rm rx}^2 = 0.01$. The path-loss constant is normalized to $\beta = 1$, with the noise variances scaled accordingly. The target reflection variance is set to $\rho_t^2=1$ and the time frame length is $T=5$. For the PASS setting, the spacing between transmit waveguides is set to $d=4$, the gap between the transmit and receive blocks is
$D_{\mathrm{gap}}=4$, and the height is $h=3$.  The service region is $x\in[1,D_{x}-1]$ and $y\in[0,y_{\max}]$, with $y_{\max}$
determined by the geometry.  We consider $K=4$ users distributed in two clusters whose locations are drawn using the
centroid-first disk model with radius $0.5$. Reported curves are averaged over 40 seeds in Fig.~\ref{fig:rateregion} and 10 seeds in Figs.~\ref{fig:convergence}, \ref{fig:sidelength}, and~\ref{fig:nelements}.

The proposed scheme is denoted by \textit{PASS (L+W)}, which alternately optimizes the precoder $\mathbf{W}$ and the pinching location. The scheme is compared against two benchmarks: \textit{Baseline~A} in which the transmit pinching elements are uniformly distributed over $[0,L_{\rm tx}]$, and $\mathbf{W}$ and receive pinching locations are optimized.
\textit{Baseline~B} which represents a fixed-location MIMO-PASS benchmark with the transmit elements being centered around $L_{\rm tx}/2$ with guided-wavelength spacing.  In this case, only $\mathbf{W}$ is optimized.


\begin{figure}[t]
\centering
\includegraphics[width=0.9\linewidth]{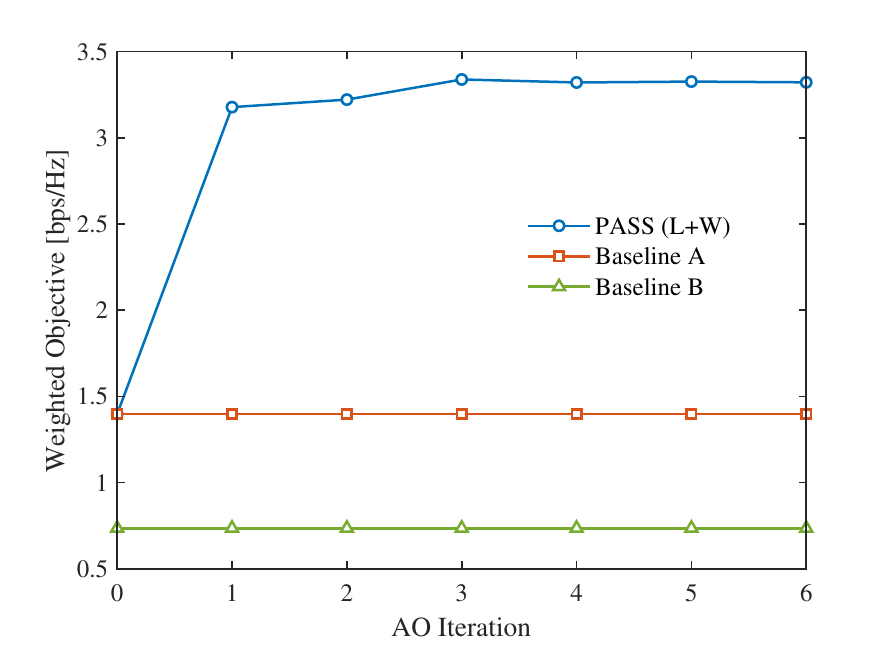}
\caption{Weighted sensing and communication rate vs. AO iteration.}
\label{fig:convergence}
\end{figure}

\subsection{Numerical Results}
\subsubsection*{Algorithm Convergence}
Fig.~\ref{fig:convergence} plots the average weighted 
objective $w_c R_c + w_s I_s^{\mathrm d}$ versus the AO iteration 
index at the balanced weighting 
$(w_c, w_s) = (0.5, 0.5)$. The proposed PASS scheme 
starts from the same uniform initialization as 
Baseline~A, performs one grouped re-initialization, and 
then refines the configuration through five AO updates. 
The two baselines produce constant trajectories, as 
they only optimize $\mathbf{W}$ for fixed element 
positions. The grouped-initialization step alone yields a 
substantial improvement, confirming the benefit of 
concentrating elements near the projection points. Two further AO updates close most of the remaining gap, after which the trajectory is essentially flat. Such rapid saturation is what makes the design usable at the system update rate; see the discussion in Section~\ref{sec:intro}. In fact, in this case, the entire AO loop must fit within one coherence interval. 
The result shows significant improvement as compared to baselines, indicating the benefit of location 
optimization in the hybrid design.

\subsubsection*{Sensing-Communication Trade-off}
The achievable region for communication-sensing rate is 
obtained by sweeping $w_s$ with $w_c = 1-w_s$ between $0$ and $1$. The result is shown in 
Fig.~\ref{fig:rateregion}. 
The markers indicate all sampled weight points, while the connected non-dominated points represent the outer front. The proposed design strictly dominates both baselines 
in the Pareto sense across the entire sweep, containing both baselines entirely. 
The proposed scheme is also nearly flat over a 
wide range of sensing levels, indicating that 
sensing requirements are met with minimal 
communication penalty.

\begin{figure}[t]
\centering
\includegraphics[width=0.9\linewidth]{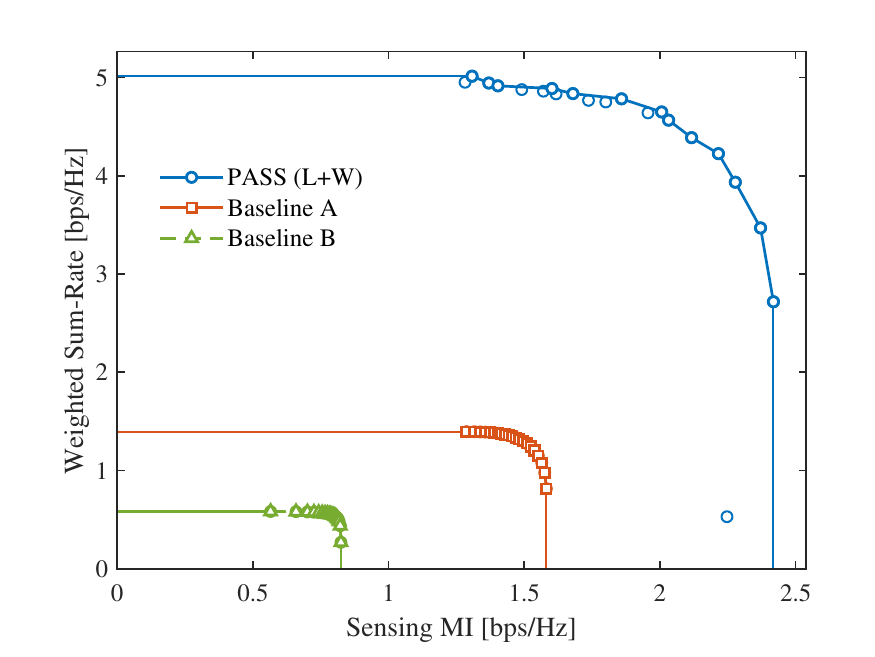}
\caption{Sensing-communication trade-off for proposed scheme and baselines.}
\label{fig:rateregion}
\end{figure}

\subsubsection*{Impacts of Environmental and System Factors}
We investigate how the weighted objective 
depends on the waveguide aperture. The transmit and 
receive lengths are jointly swept while the element counts are held fixed at $N_{\mathrm{tx}} = N_{\mathrm{rx}} = 20$. 
Fig.~\ref{fig:sidelength} shows the results. The fixed baselines degrade as the aperture grows with a fixed number of elements, while PASS remains robust by reallocating elements near the scene-dependent projection points. This scalability advantage is a direct consequence of treating the element locations as a design variable rather than a fixed geometric constant.
\begin{figure}[t]
\centering
\includegraphics[width=0.9\linewidth]{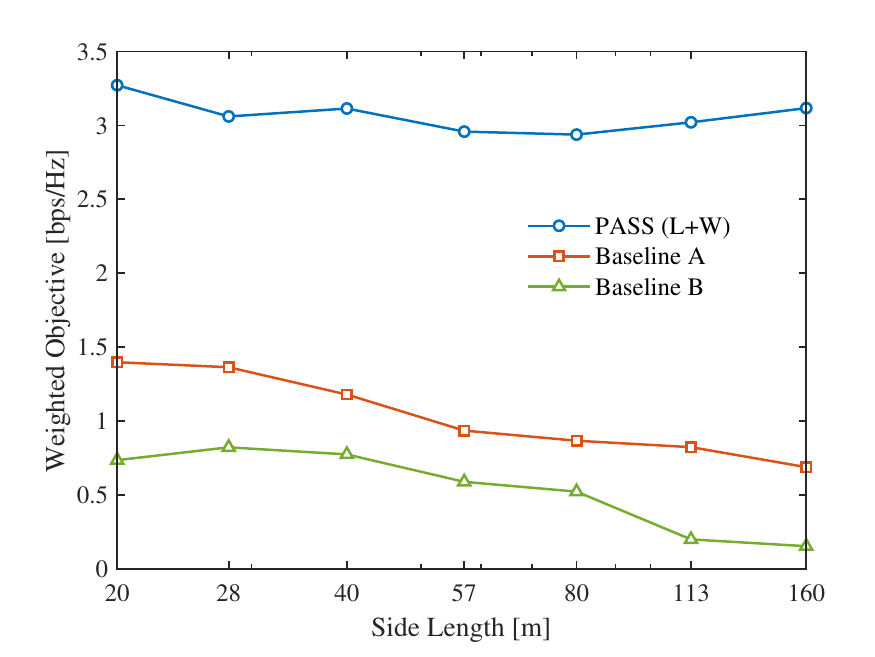}
\caption{Weighted sensing and communication rate vs. side length $D_x$.}
\label{fig:sidelength}
\end{figure}


Finally, we study the effect of the element budget on 
the final weighted objective. The waveguide lengths are 
held fixed at $L^{\mathrm{tx}} = L^{\mathrm{rx}} = 20$, 
and the element counts are jointly varied. 
Fig.~\ref{fig:nelements} shows the effect of the element budget. PASS benefits from additional elements, up to a saturation point, because the allocation variables direct them toward useful user or target projections, until the growing physical extent of a group offsets the array gain.

\begin{figure}[t]
\centering
\includegraphics[width=0.9\linewidth]{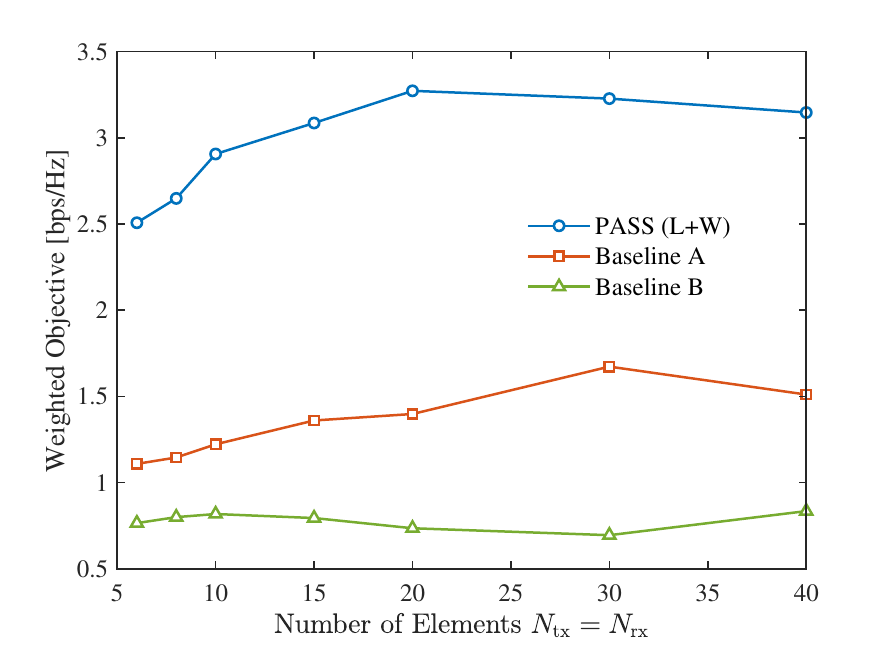}
\caption{Weighted sensing and communication rate vs. the number of elements.}
\label{fig:nelements}
\end{figure}

\section{Conclusion}
\label{sec:conclusion}
We proposed a hybrid design for PASS-assisted ISAC systems, in which the digital precoder and pinching locations are jointly designed to maximize both communication and sensing rates. The problem was formulated as a multi-objective optimization. Using the scalarization technique, the problem was reformulated as a weighted sensing and communication sum rate maximization. Due to its non-convexity, the resulting design problem was solved iteratively via the AO scheme. The marginal precoding subproblem was transformed into an equivalent WMMSE task, whose closed-form update rule was derived via rank-one update. The pinching location optimization task was further addressed via user clustering trick, which transformed the original problem into a classical weight allocation task on a product of probability simplexes. Numerical results validate the superiority of the proposed scheme 
as compared with the basic PASS-assisted and conventional fixed-antenna 
baselines.  
Our findings indicate that the spatial degrees of freedom offered by PASS can be efficiently deployed in practical ISAC settings. 


%




\ifCLASSOPTIONcaptionsoff
  \newpage
\fi



\bibliographystyle{IEEEtran}
\bibliography{refs}
\end{document}